\documentclass[11pt]{article}
\usepackage{epsfig}
\usepackage{amsmath}
\usepackage{amssymb}
\usepackage{boxedminipage}
\usepackage{algorithm}

\usepackage{tikz}
\usepackage{tkz-berge}
\usepackage{tkz-graph}

\newcommand {\be}[1]{\begin{equation}\label{#1}}
\newcommand {\ee}{\end{equation}}
\newcommand {\bea}{\begin{eqnarray}}
\newcommand {\eea}{\end{eqnarray}}

\newcommand{\qed}{\hfill \rule{7pt}{7pt}}

\newtheorem{theorem}{Theorem}
\newtheorem{lemma}{Lemma}

\begin{document}

\title{A note on the rationing of divisible and indivisible goods in a general network}

	\author{
Shyam Chandramouli\thanks{IEOR Department, Columbia University, New York,
NY; {\tt sc3102@columbia.edu}} and
Jay\
Sethuraman\thanks{IEOR Department, Columbia University, New York,
NY;
{\tt jay@ieor.columbia.edu}}}


\maketitle

\begin{abstract}
The study of matching theory has gained importance recently with applications in Kidney Exchange, House Allocation, School Choice etc. The general theme of these problems is to allocate goods in a fair manner amongst participating agents. The agents generally have a unit supply/demand of a good that they want to exchange with other agents. On the other hand, Bochet et al.~\cite{bim,bims} study a more general version of the problem where they allow for agents to have arbitrary number of divisible goods to be rationed to other agents in the network. In this current work, our main focus is on non-bipartite networks where agents have arbitrary units of a homogeneous indivisible good that they want to exchange with their neighbors. Our aim is to develop mechanisms that would identify a fair and strategyproof allocation for the agents in the network. Thus, we generalize the kidney exchange problem to that of a network with arbitrary capacity of available goods. Our main idea is that this problem and a couple of other related versions of non-bipartite fair allocation problem can be suitablly transformed to one of fair allocations on bipartite networks for which we know of well studied fair allocation mechanisms.
\end{abstract}

\newpage

\section{Introduction}

The study of fair allocation on economic networks has gained a lot of importance in recent years with growing applications in many public policy domains like fair exchange of kidneys among patients~\cite{roth2005pairwise}, matching students to public schools~\cite{abdulkadirouglu2003school},  matching cadets in rotc~\cite{sonmezrotc} etc. 
The common theme in all these problems is that participating agents are in supply/demand of a unit indivisible good and the set of agents with whom they can share/recieve/supply this good 
is modeled by a link. These problems identify fair and strategyproof allocation mechanisms 
for agents in the network and thus is in very similar spirit of the classical marriage problem of Gale and Shapley~\cite{gale1962college}.

On the other hand, Bochet et al.~\cite{bim,bims} study the problem of fair division of a maximum flow in a capacitated bipartite network. This model
generalizes and studies the exchange of divisible goods on a economic networks where agents have arbitrary supply/demand constraints. Note that, the aforementioned problems were typically unit supply/demand model. 
The common feature in both these research is that
the associated market  is moneyless, so that fairness is achieved by
equalizing the allocation {\em as much as possible}. This last caveat is to account for additional
considerations, such as Pareto efficiency and strategyproofness, that may be part of the planner's objective.

A well-studied special case of the Bochet et al.~\cite{bim,bims} problem is that of allocating a single resource (or allocating the
resource available at a single location) amongst a set of agents with varying (objectively verifiable) claims on it. 
This is the special case when there is a single supply node that is connected to every one of the demand
nodes in the network by an arc of large-enough capacity. If the
sum of the claims of the agents exceeds the amount of the resource available, the problem
is a standard rationing problem (studied in the literature as ``bankruptcy'' problems or ``claims''
problems). There is an extensive literature devoted to such problems that has resulted in a
thorough understanding of many natural methods including the {\em proportional} method,
the {\em uniform gains} method, and the {\em uniform losses} method. A different view of
this special case is that of allocating a single resource amongst agents with single-peaked
preferences over their net consumption. Under this view, studied by Sprumont~\cite{spr},
Thomson~\cite{thomson} and many others, the goal is to design a mechanism for allocating
the resource that satisfies appealing efficiency and equity properties, while also eliciting
the preferences of the agents truthfully. The {\em uniform rule}, which is essentially an adaptation 
of the uniform gains method applied to the reported peaks of the agents, occupies a central
position in this literature: it is strategy-proof (in fact, group strategy-proof), and finds an
envy-free allocation that Lorenz dominates every other efficient allocation; furthermore,
this rule is also {\em consistent}. A natural two-sided version of Sprumont's model has agents initially 
endowed with some amount of the resource, so that agents now fall into two categories: someone
endowed with less than her peak is a potential demander, whereas someone endowed with more than
her peak is a potential supplier. The simultaneous presence
of demanders and suppliers creates an opportunity to trade, and the obvious adaptation of
the uniform rule gives their peak consumption to agents on the short
side of the market, while those on the long side are uniformly rationed (see 
\cite{KPS}, \cite{BJ}). This is again equivalent to a standard rationing problem because
the nodes on the short side of the market can be collapsed to a single node.
The model we consider generalizes this by assuming that the resource can only be transferred 
between certain pairs of agents. Such constraints are typically logistical (which supplier can
reach which demander in an emergency situation, which worker can
handle which job request), but could be subjective as well (as when a
hospital chooses to refuse a new patient by declaring red
status). This complicates
the analysis of efficient (Pareto optimal) allocations, because short demand
and short supply typically coexist in the same market. 

As mentioned earlier, Bochet et al.~\cite{bim, bims}. 
work with a bipartite network in both papers and assume that each node is populated
by an agent with single-peaked preferences over his consumption of the resource: thus,
each supply node has an ``ideal'' supply (its peak) quantity, and each demand node has
an ideal demand. These preferences are assumed to be private information, and Bochet
et al.~\cite{bim, bims} propose a clearinghouse mechanism that collects from each agent
only their ``peaks'' and picks Pareto-optimal transfers with respect to the reported peaks.
Further, they show that their mechanism is strategy-proof in the sense that it is a dominant
strategy for each agent to report their peaks truthfully. While the models in the two papers
are very similar, there is also a critical difference: in~\cite{bims}, the authors require that
no agent be allowed to send or receive any more than their peaks, whereas in~\cite{bim} the
authors assume that the {\em demands} must be satisfied exactly (and so some supply nodes
will have to send more than their peak amounts). The mechanism of Bochet et al.---the {\em egalitarian}
mechanism---generalizes
the uniform rule, and finds an allocation that Lorenz dominates all Pareto efficient allocations.
Later, Chandramouli and Sethuraman~\cite{cs,cs2} show that the egalitarian mechanism is in fact strongly invariant,
peak and link group strategyproof: it is a dominant strategy for any group of agents (suppliers or demanders) 
to report their peaks truthfully. 
Szwagrzak~\cite{sz1} studies the property of contraction invariance of an allocation rule: when 
the set of feasible allocations contracts such that the optimal allocation is still 
in this smaller set, then the allocation rule should continue to select the same allocation.
He shows that the egalitarian rule is contraction invariant. 
These results suggest that the egalitarian mechanism may be the
correct generalization of the uniform rule to the network setting.

Our research is motivated by these results. From the results above, the generalized model of allocation of divisible goods on bipartite networks is well understood. On the contrary, there has not been a much focus in the litreature on fair mechanisms on non-bipartite networks. Our contribution can be viewed in some sense a generalization of identifying a fair maximum matching  on general networks to 
that of identifying a fair maximum b-matching problem on non-bipartite networks. Our contributions parallel the contribution of Bochet et al.~\cite{bim,bims} in generalizing unit capacity models. 

Specifically, we are given a non-bipartite network $G = (N, E)$, and we think of $N$ as the set of agents involved in this network. Each arc $(i,j) \in E$ connects two participating agents and has capacity $u_{ij} \geq 0$. There is a single commodity (the resource) that is available at
each node and needs to be exchanged with the neighbors on the network: we assume that  node $i$ has
$b_i$ units of the resource. The capacity of an arc $(i,j)$
is interpreted as an upper bound on the direct transfer from supply node $i$ to demand node $j$. The agents derive utility whenever they exchange a good with their neighbors. 
 The
goal is to find a maximum "fair" exchange among participating agents , while also respecting
the capacity constraints on the arcs. We would describe in more detail our model in the later sections. We describe below the connections of this model with litreature.

\begin{itemize}
\item \textbf{Connection to Kidney Exchange Problem}: The kidney exchange problem of~\cite{roth2005pairwise} and the subsequent heaavy litereature, study a unit exchange problem between patients who are connected to their compatible donors through links in the network. Each agent needs exactly one unit of good (in this case, kidney). Since, not every patient is compatible with everyone else, we are naturally poised with the problem of not able to match every patient with his compatible donor, hence the goal of the study is to identify maximum number of matches as possible. In a combinatorics view, this problem can be viewed as one of finding a maximum matching on a non-bipartite network. When the set of maximum matchings is not unique, the mechanism designer is forced to pick one such maximum matching. From an economics perspective, the mechanism designer would like to satisfy certain fairness objectives like Lorenz Dominance, Pareto Optimality etc. and also incentive compatible conditions. Roth et al.~\cite{roth2005pairwise} construct the "Egalitarian Mechanism" - which is a randomized lottery mechanism over the set of maximum matchings. They identify a fair and strategyproof matching. As described above, our problem is one in which each agent $i$, in the networks has arbitrary units $b_{i}$ and derives utility with every exchange of that good with his/her neighbors in the network. Since the goods are indivisible, our model generalizes the kidney exchange problem to multiple demands and from a combinatorics perspective, generalizes to one one of identifying a maximum b-matching on a non-bipartite network.

\item \textbf{Ordinal transportation problem}: Ordinal transportation problem of Balinski and Yu~\cite{baiou2002erratum}, study the problem of identifying a stable b-matching on bipartite networks. We differ from their study in many ways. Our model is more general non-bipartite networks, our preference structure is not ordinal (agents are indifferent from whom they recieve the goods from, the utility structure is single peaked) and we are more interested in fairness notions like pareto optimality and envyfreeness versus the notion of stability studied by Balinski. 
\end{itemize}

\textbf{Our Contributions};
Our main contribution in this paper is a generalized egalitarian mechanism that is Lorenz dominant, pareto optimal, envyfree and strategyproof for agents participating in a non-bipartite network. We identify this in both the divisible and indivisible goods case.  We do so by transforming each of these non-bipartite networks to a suitable bipartite network. We solve the fractional matching problem on this bipartite network and show how to construct a lottery over integral b-matchings in the original network.

The rest of the paper is organized as follows: in Section~\ref{s:model1} we consider the divisble goods case and in Section ~\ref{s:model2} we consider the indivisble goods case. 
Finally in Section ~\ref{s:extensions}, we study extensions to problems with capacities, weighted matchings and general s-t networks, non-pairwise exchange

\section*{Definitions:}
Here we describe some of the common notions that are used in this paper. We are describing it explicitly here because it is the same terminology that we use across our models. \\

\noindent \underline{\textbf{Pareto Optimality:}} 
A feasible net transfer $x$ as defined in the previous section is Pareto Optimal if there
is no other allocation $x'$ such that every agent is weakly better off and atleast
one agent is strictly better off in it. In mathematical terms, if $R_{i}$ and $I_{i}$
denote the preference and indifference relations respectively for agent $i$, then 
\begin{equation}
 \{ \forall \hspace{2mm} i: \hspace{2mm} x_{i}'R_{i}x_{i} \hspace{2mm} \} \implies \{ \forall \hspace{2mm} i: \hspace{2mm} x_{i}'I_{i}x_{i}  \}
\end{equation}

\noindent \underline{\textbf{Lorenz Dominance:}}
A solution is Lorenz Dominant inside the
pareto optimal set if for any $z\in \mathbb{R}^{V}$, write $z^{\ast }$ for
the \textit{order statistics }of $z$, obtained by rearranging the
coordinates of $z$ in increasing order. For $z,w\in \mathbb{R}^{V}$ , we say
that $z$ \textit{Lorenz dominates} $w$, written $z$ $LD$ $w$, if for all $%
k,1\leq k\leq n$ 
\begin{equation*}
\sum_{a=1}^{k}z^{\ast a}\geq \sum_{a=1}^{k}w^{\ast a}
\end{equation*}%
Lorenz dominance is a partial ordering, so not every set, even convex and
compact, admits a Lorenz dominant element. On the other hand, in a convex
set $A$ there can be at most one Lorenz dominant element. The appeal of a
Lorenz dominant element in $A$ is that it maximizes over $A$ \textit{any}
symmetric and concave collective utility function.\smallskip \\

\noindent \underline{\textbf{No Envy:}} A rule $x \in  \mathcal{F}(G,b,u)$ satisfies \emph{No Envy} if for any preference profile $R \in \mathcal{R}^{S \cup D}$
and $i,j \in
S$ such that $x_{j}P_{i}x_{i}$, there exists no $x'$ such that 
\begin{eqnarray}
&x_{k} = x'_{k} \ \text{for all} \ k \in S \backslash \{ i,j \} ;  \text{for all} \ l \in D \ \text{and} 
&x'_{i}P_{i}x_{i} 
\end{eqnarray}

\noindent \textbf{\underline{Strategyproof}:} A rule $x$ on $(G,b,u)$ is \emph{strategyproof} if 
for all $R \in R^{S \cup D}$, $i \in V$ and $R_{i}' \in \mathcal{R}$
\begin{equation}
 x_{i}(R)R_{i}x_{i}(R'_{i},R_{-i}) \hspace{2mm} 
\end{equation}

\section{Model 1: With preference indifference, divisible goods}
\label{s:model1}
In this section we consider the version of the problem where
the nodes of the network are populated by agents. 
Each node has a specific number of units of a particular good that they can exchange with the agents that they are connected with.
Thus, our problem becomes one of exchanging a single commodity
among the set of agents $V$ using the set $E$ of edges.

An exchange of the commodity among the agents is 
realized by a b-matching $f$,
which specifies the amount of
the commodity exchanged among the agents $i$ and $j$ using
the edge $(i,j) \in E$. The flow $f$ induces an allocation
vector for each agent as follows:
\begin{equation}
\text{for all }i\in V:x_{i}(f)=\sum_{j\in N(i)}f_{ij}; \label{1}
\end{equation}%
As we shall see in a moment, agents only care about
their {\em net} transfers, and not on how these transfers are distributed
across the agents on the other side.

In the following sections, we assume that the peaks of the agents are fixed, and focus our attention on mechanisms
that elicit preferences from the agents on their connectivity and maps the reported connectivity for each preference
profile to a unique b-matching.

To summarize: agents arrive to the market with a specified number of divisible goods and report the agents with whom they can have exchanges; They are preference homogeneous in the sense that they are indifferent between whom they exchange.  the allocation rule is applied to the graph $G$ with edge-capacities $u$, and the data $b$. Our focus will be on mechanisms in
which no agent has an incentive to misreport his compatible neighbors and also efficient.

\subsection*{Allocation Rules}

We define an allocation on the edges as $f_{ab} \hspace{2mm}, ab \in E$ as feasible, if $f_{ab} = f_{ba} \hspace{2mm} \forall ab \in E$ and the flow induced on the nodes by $f$ (as $x$ defined earlier) and for feasibility, $x_{a} \leq b_{a} \hspace{2mm} \forall a\in V$. Any rule which picks an allocation from this set is called a feasible allocation rule. In the rest of the section, we discuss the egalitarian rule for this model.

Given a network $(G,N,E)$, we make the following transformation to construct the bipartite
network $(G_{b},V_{b},E_{b})$. We represent $V_{b} = A \cup B$ where $A = V$ and $B=V$. are either sides of the network.  There is an edge between agent $i \in A$ and $j \in B$ only
if there is an edge $ij$ in the given network $G$. Connect the agents in $A$ to a supply node $s$ and the agents in $B$ to a sink node $t$. We refer to as a flow, any feasible
shipment of goods from the source node to the sink node.


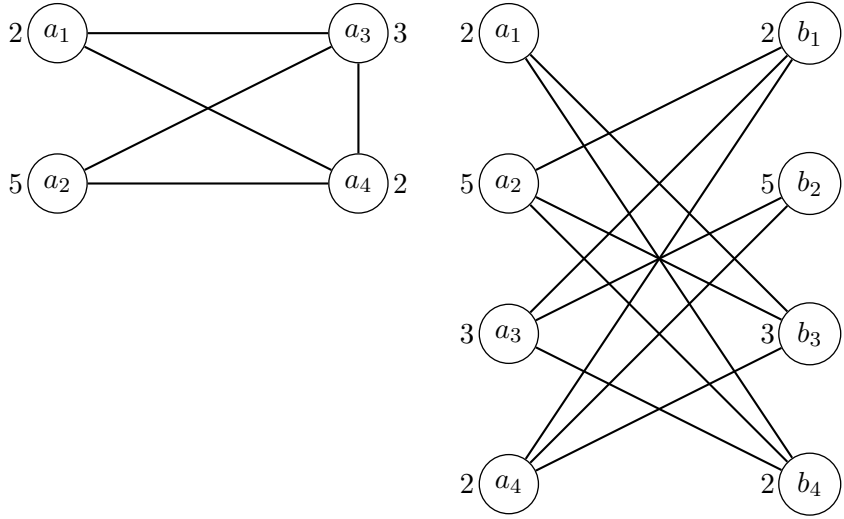
\begin{figure}
 \begin{center}
  \begin{tikzpicture}
   \GraphInit[vstyle=Normal]



	\Vertex[LabelOut,Lpos=180, x=0 ,y=0]{5}
 	\Vertex[x=0 ,y=0,Math,L=a_2]{a2}

 	\Vertex[LabelOut,Lpos=180, x=0 ,y=2]{2}
	\Vertex[x=0 ,y=2,Math,L=a_1]{a1}


	\Vertex[LabelOut,Lpos=0, x=4,y=0]{2}
 	\Vertex[x=4 ,y=0,Math,L=a_4]{a4}

	\Vertex[LabelOut,Lpos=0, x=4 ,y=2]{3}
 	\Vertex[x=4 ,y=2,Math,L=a_3]{a3}

     \Edges(a1,a3)
     \Edges(a1,a4)
     \Edges(a2,a4)
     \Edges(a2,a3)
     \Edges(a4,a3)


  	\Vertex[LabelOut,Lpos=180, x=6 ,y=2]{2}
	\Vertex[x=6 ,y=2,Math,L=a_1]{a1}

	\Vertex[LabelOut,Lpos=180, x=6 ,y=0]{5}
	\Vertex[x=6 ,y=0,Math,L=a_2]{a2}

	\Vertex[LabelOut,Lpos=180, x=6 ,y=-2]{3}
	\Vertex[x=6 ,y=-2,Math,L=a_3]{a3}

	\Vertex[LabelOut,Lpos=180, x=6 ,y=-4]{2}
	\Vertex[x=6 ,y=-4,Math,L=a_4]{a4}


\Vertex[LabelOut,Lpos=180, x=10 ,y=2]{2}
	\Vertex[x=10  ,y=2,Math,L=b_1]{b1}

	\Vertex[LabelOut,Lpos=180, x=10  ,y=0]{5}
	\Vertex[x=10  ,y=0,Math,L=b_2]{b2}

	\Vertex[LabelOut,Lpos=180, x=10  ,y=-2]{3}
	\Vertex[x=10  ,y=-2,Math,L=b_3]{b3}

	\Vertex[LabelOut,Lpos=180, x=10  ,y=-4]{2}
	\Vertex[x=10  ,y=-4,Math,L=b_4]{b4}

 \Edges(a1,b3)
 \Edges(a1,b4)

\Edges(a2,b1)
 \Edges(a2,b3)
 \Edges(a2,b4)

\Edges(a3,b1)
 \Edges(a3,b2)
 \Edges(a3,b4)

 \Edges(a4,b1)
 \Edges(a4,b2)
 \Edges(a4,b3)

	\end{tikzpicture}
   \caption{Transformation to bipartite network} 
   \label{fig:model_div}
  \end{center}
\end{figure}

\begin{lemma}
Every feasible flow in the modified network $G_{b}$ corresponds to a feasible 
allocation (exchange) in the original network $G$. BIMS's Egalitarian rule on the modified network $G_{b}$ results in a pareto optimal, lorenz dominance, envy free and peak strategy proof allocation for the agents in the original network $G$.
\end{lemma}

\begin{proof}
Consider a solution $(y_{a},y_{b},g_{ab})$ for the agents in the modified bipartite network where $y_{a},y_{b}$ refers to an allocation for agent $a \in A$ and $b\in B$ respectively. 
$g_{ab}$ is the flow on the edge $ab$. We will try to show that every feasible solution in the modified bipartite network maps to a solution in the original network and vice-versa. Once we have that, the second part of the statement follows (As Bochet et al.~\cite{bim,bims} have established the fairness properties of egalitarian mechanism in bipartite networks). 

Firstly, consider a feasible exchange in the original problem. Denote the solution by $(x_{a},f_{ab})$ for agent $a,b$ in the network. $x,f$ are node and edge allocations respectively. In the modified network, set $g_{a_{i}b_{j}} = g_{a_{j}b_{i}} = f_{a_{i}a_{j}} \hspace{2mm} \forall i,j \in V$. Clearly, this solution is feasible in the modified network and $y_{a_{i}} = y_{b_{i}} = x_{a_{i}} \hspace{2mm} \forall i \in V$.

Now, consider a solution $(y_{a},y_{b},g_{ab})$ in the modified bipartite network. Construct the following solution, $f_{a_{i}a_{j}} = (g_{a_{j}b_{i}} + g_{a_{i}b_{j}})/2 \hspace{2mm} \forall i,j \in V$. This also defines a node allocation in the original network with $x_{a_{i}} = (y_{a_{i}} + y_{b_{i}})/2 \hspace{2mm} \forall i \in V $. This is a feasible allocation since, $f_{ij} = f_{ji} \hspace{2mm} \forall ij \in E$ and $x_{i} \leq b_{i} \hspace{2mm} \forall i \in V$.
\qed
\end{proof}

Hence, the case of exchanging divisible goods in an economic network is done by a simple transformation into a bipartite network. The mechanism is also peak strategyproof if the preference profiles are private information of the agents. In the next section, we discuss the case of exchanging indivisible goods among agents in a network. We show that there is no mechanism that is peak strategyproof but we design a mechanism which is egalitarian in nature and retains many other attractive properties.

\newpage


\section{Model 2: With preference indifference, Indivisible goods}
\label{s:model2}
In this section, we focus on allocation rules, which allows only integer allocations to the agents in the network. To summarize: agents arrive to the market with a specified number of indivisible goods and report 
the agents with whom they can have exchanges; They are preference homogeneous in the sense that 
they are indifferent between whom they exchange.  
the allocation rule is applied to
the graph $G$ with edge-capacities $u$, and the data $b$. Our focus will be on mechanisms in which no agent has an incentive to misreport his connectivity.

At this juncture, we would like to highlight an important difference with that of the
bipartite network structure here. The egalitarian rule of Bochet et al.~\cite{bim,bims} constructs a fair allocation for agents in a bipartite network with divisible goods. This rule is a generalization of the
well known Sprumont's~\cite{spr} uniform rule. On the other hand, when we are allocating a set
of indivisible goods among agents, Klaus et al.~\cite{ehlers2003probabilistic} proposed the probablistic uniform rule 
in the single agent model. Their main contribution is the fact that there is no net utility loss for 
agents in both divisible/indivisible models. The fractional part of an agents allocation (expected utility from this mechanism) is the probability with which he has a claim on an extra unit of good. 

Klaus et al.~\cite{ehlers2003probabilistic} mechanism is based on a simple idea that at each bottleneck point of the 
Sprumont's model, we randomize over all possible feasible allocations. Similarly, we define an egalitarian mechanism for indivisible goods. The idea is to randomize
over all possible feasible flows at each bottleneck. The utility for the agents in the divisible goods case is same as the expeced utility in the indivisible goods case. 

This is not the case in non-bipartite networks. The net utility (sum total of utilities for all agents) is not the same in divisible/indivisible goods case. For a simple example, consider
a traingular network with nodes $(a,b,c)$ each with peak = 1; In the divisible goods case, 
a pareto allocation would assign 1 unit to each node (0.5 units on each edge).  The net social utility is 3units in this case. If the goods are indivisible, only 1 of the 3 edges can be picked in any maximal allocation (pareto). The net social utility is 2units. In a fair allocation, we would pick each edge with equal probability, giving an allocation of $2/3$ for each agent. Hence, there is a clear gap in the achievable utility between divisible and indivisible goods
case in non-bipartite networks.

Hence, in this section we aim to identify an egalitarian mechanism for this model where the agents arrive with indivisible goods. We reduce it to a suitable bipartite network and apply the well-known fair allocation rules on this modified network.

Extensions of bipartite allocation rules: We note another property of the mechanisms which would be different when we move to non-bipartite networks. "Extension" is the property that when we restrict ourselves to bipartite networks, the mechanisms that we construct should boil down to the familiar fair allocation rules. Due to the "utility gap" we discussed earlier, the divisible and indivisible rules on non-bipartite networks would reduce to their respective counterparts on bipartite networks. Though on bipartite networks the rules look similar, they need to be adjusted for the "utility gap" when we move to non-bipartite networks. We would discuss this aspect in further detail in later sections

\subsection*{\underline{b-matchings}} 
b-matchings is a generalized notion of a matching type problems on networks. 
Given a network $G$ and a positive number $b_{i}$ for each node $i \in V$ and capacity 
$u_{e}$ for each edge $e \in E$
We define the 
\emph{u-capacitated b-matching} or (b,u) matching is a vector $x$ such that 
\begin{eqnarray}
\label{defn:bmatching}
&x_{i}(f)  \leq b_{i} \hspace{2mm} \text{for all} \hspace{2mm} i \in V \\
&0 \leq f_{ij} \leq u_{ij} \hspace{2mm} \text{for all} \hspace{2mm} ij \in E
\end{eqnarray}

In the case when all the $u_{e} = \infty$ and $b_{i} = 1$, we arrive at the problem 
of finding a matching in network $G$.  The set of all $x$ which satisfies [\ref{defn:bmatching}] is the
set of all feasible b-matchings. Clearly, a given graph $G$ can have more than 
one b-matching. From an operations research as well as economics point of view we will 
be interested in finding the \emph{maximum b-matching} among all feasible b-matchings. Here, 
the maximum b-matching is the one with the highest $\sum_{i} x_{i}$ among all the feasible b-matchings.

A maximum weight u-capacitated b-matching problem can be solved in strongly polynomial time by a reduction to the maximum weight b-matching problem. Following is a linear programming formulation of the maximum weight b-matching problem:

\begin{eqnarray}
&\text{Max} \hspace{2mm} \sum_{i} \sum_{j} X_{ij} \\
 & \text{subject to} \\
& \sum_{j \in N(i)} X_{ij} \leq b_{i} \hspace{2mm} \forall i \in V(G) \\
& X_{ij} \in Z
\end{eqnarray}

The maximum weight u-capacitated b-matching problem would encode the additional
constraint $X_{ij} \leq u_{ij} \hspace{2mm} \forall ij \in E(G)$. In this section, we assume $u_{ij} = \infty$ but we discuss the capacitated case in later sections.

Algorithms for finding a maximum weight b-matching: Cook and Cunningham~\cite{cook2011combinatorial} discuss
polynomial algorithms to solve the b-matching problem. We briefly discuss about the b-matching matroid.  This discussion follows from the well-known matching matroid formulation. For details on matching matroid, refer to cook~\cite{cook2011combinatorial}.

For an undirected graph $G(V,E)$ define, $\mathcal{I} = \{x_{vec} | x_{vec} \hspace{2mm} \text{is the node allocation from some b-matching of G}\}$ where $x_{vec} \in R^{V}$, the $i^{th}$ component of the vector, $x_{vec}^{i} = x_{i}$ where $x_{i}$ is the number of times node $i$ is matched in that particular b-matching. Then the pair, $(V,\mathcal{I}):=M$ is a matroid. We will refer to it as the b-matching matroid. A base in $M$ is a vector $x_{vec}$ that is generated by some maximum b-matching of $G$. 

The "rank" function : $2^{V} \rightarrow Z^{+} \cup \{0\}$  of the matroid $M$ is defined to be 
rank(S) = $\max_{I \subset S, I \in \mathcal{I}} |I| \hspace{2mm} \forall \hspace{2mm} S \subseteq V$. The rank function is submodular and in our context rank(S) can be interpreted as the size of a maximum b-matching in the set S. $P_{rank}$ refers to the polymatroid associated with this particular rank function and is given by \\
\begin{equation}
P_{rank} = \{x \in R^{V}|x\geq 0, x(S) \leq f(S), \hspace{2mm} \forall \hspace{2mm} S \subseteq V\} 
\end{equation}
where $x(S) = \sum_{i \in S} x_{i}$

From an Operations Research perspective, we have obtained an optimal solution to the maximum weight b-matching problem. But the aforementioned linear program can have many solutions. From an Economics perspective, we would like to pick a particular maximum weight b-matching from the optimal set such that the utility of the agents satisfy fairness constraints. Inoder to find such an allocation, we would understand the structure of maximum weight b-matchings for the uncapacitated problem in the rest of this section. Towards the end, we analyze the case with capacities.

\subsection*{\underline{Pareto Optimality}}
A feasible allocation $x \in \mathcal{A}(G)$ is \emph{Pareto Optimal} if for any other $x' \in \mathcal{A}(G)$ we have 
\begin{equation}
\{ \text{for all i:} \hspace{2mm}  x_{i}'R_{i}x_{i}\} \implies \{  \text{for all i:} \hspace{2mm}  x_{i}'I_{i}x_{i}  \}
\end{equation}

i.e. there is no other allocation $x'$ that is weakly better for every agent and strictly better for at least one agent in the allocation $x$. 

\begin{lemma}
The set of pareto optimal allocations is equivalent to the set of maximum weight b-matchings
\end{lemma}

\textbf{Proof:} If $x$ denotes the utiliy profile of the agents in the network, then the size of the b-matching that produced this utility profile is $x/2$. As each exchange between 2 agents on a network adds an utility of one to each of those agents.

Suppose, we have an exchange of goods in the network that is produced by a maximum b-matching resulting in the allocation profile $x$. Since, it is a maximum b-matching, we cannot have more feasible exchanges in this network. There is no alternating path of odd length connecting two unsaturated nodes. Any other path of even length from an unsaturated node $a$, to another unsaturated node $b$ will result in another maximum b-matching with the utility of node $b$ strictly below the current solution. Hence, picking any maximum b-matching to the problem should result in an pareto optimal allocation for the agents.

Now, suppose $x$ is a pareto optimal allocation for the agents in the network, then there exists no allocation $y$ in the network such that every agent is weakly better under allocation $y$, and there is atleast one agent strictly better off under $y$. This implies, there is no alternating path of odd length connecting two unsaturated nodes. This implies the given solution is a maximum b-matching. \qed

The following generalization of Gallai and Edmonds theorem further helps us understand the structure of pareto optimal solutions

\subsection*{\underline{Gallai - Edmonds Decomposition}}

Gallai and Edmonds decompose any given network into smaller sub-components namely, over demanded, under demanded and pefectly demanded components. Bochet et al.~\cite{bim,bims} use this decmposition to establish the structure of pareto optimal allocations in bipartite networks. Roth et al.~\cite{roth2005pairwise} establish the structure of Pareto-efficient pairwise matchings using the Gallai-Edmonds decomposition (GED). We generalize the GED idea to our problem where the nodes have arbitrary peaks associated with them. Hence, we obtain a characterization of pareto efficient b-matchings in this problem. $M$ is the set of maximum b-matchings and let $\mu$ denote an arbitrary matching in the set $M$, $\mu(i)$ denotes the matched neighbors of agent $i$ in matching $\mu$. Setting up the notation
partition $V$ as $\{V^{U},V^{O},V^{P}\}$ such that 

\begin{eqnarray}
V^{U} &=& \{  i \in V: \exists \mu \in M \hspace{2mm} s.t. \hspace{2mm} i \in \mu(i)   \} \\
V^{O} &=& \{ i \in V\backslash V^{U}: \exists \text{set of agents} \hspace{2mm} \tilde{j} \in V^{U} s.t.  \sum_{j\in\tilde{j}} n_{i,j} = b_{i} \}, \\
V^{P} &=& V \backslash (V^U \cup V^O)
\end{eqnarray}

where $n_{i,j} = 1$ if $i,j$ is matched in a feasible solution.
$V^{U}$ is the set of agents unmatched ($x_{i} < b_{i}$) in atleast one maximum b-matching. 
$V^{O}$ is the set of agents completely matched ($x_{i} = b_{i}$) in every maximum b-matching and have atleast 1 neighbor in $V^{U}$. $V^{P}$ is the set of agents who are again perfectly matched but does not have a link with any agent in $V^{U}$.

\begin{figure}
\begin{center}
\begin{tikzpicture}

\Vertex[x=0,y=0]{$s_1$}

\Vertex[x=1,y=-2]{$s_2$}

\Vertex[x=2,y=0]{$s_3$}

\Vertex[x=4,y=0]{$s_4$}

\Vertex[x=5,y=0]{$s_5$}

\Vertex[x=4,y=-3]{$s_6$}

\Vertex[x=6,y=-3]{$s_7$}

\Vertex[x=7,y=0]{$s_8$}

\Edges($s_1$,$s_2$)
\Edges($s_2$,$s_3$)
\Edges($s_3$,$s_1$)
\Edges($s_6$,$s_2$)
\Edges($s_6$,$s_4$)
\Edges($s_6$,$s_5$)
\Edges($s_7$,$s_8$)

\node at (-.5,0.5){$2$};
\node at (1.5,0.5){$2$};
\node at (0.5,-1.5){$3$};

\node at (3.5,.5){$4$};
\node at (4.5,.5){$4$};
\node at (3.5,-3.5){$5$};
\node at (5.5,-3.5){$2$};
\node at (6.5,.5){$4$};

\Vertex[x=0,y=-5]{$s_{9}$}
\Vertex[x=0,y=-7]{$s_{10}$}
\Vertex[x=2,y=-7]{$s_{11}$}
\Vertex[x=2,y=-5]{$s_{12}$}
\Edges($s_{9}$,$s_{10}$)
\Edges($s_{10}$,$s_{11}$)
\Edges($s_{11}$,$s_{12}$)
\Edges($s_{12}$,$s_{9}$)
\Edges($s_{9}$,$s_{11}$)
\Edges($s_{10}$,$s_{12}$)
\node at (-0.5,-4.5){$2$};
\node at (-0.5,-6.5){$2$};
\node at (1.5,-6.5){$2$};
\node at (1.5,-4.5){$2$};

\Vertex[x=5,y=-5]{$s_{13}$}
\Vertex[x=4,y=-7]{$s_{14}$}
\Vertex[x=6,y=-7]{$s_{15}$}
\Edges($s_{13}$,$s_{14}$)
\Edges($s_{15}$,$s_{14}$)
\Edges($s_{13}$,$s_{15}$)
\node at (4.5,-4.5){$2$};
\node at (3.5,-6.5){$2$};
\node at (5.5,-6.5){$2$};

\Edges($s_{13}$,$s_7$)
\Edges($s_{12}$,$s_6$)
\end{tikzpicture}
\end{center}
\caption{GED of a non-bipartite network with arbitrary peaks; In this figure, $\{s_{6},s_{7}\} \in V^{O}, \{s_{1},s_{2},s_{3},s_{4},s_{5},s_{8}\} \in V^{U}, \{s_{9},s_{10},s_{11},s_{12},s_{13},s_{14},s_{15}\} \in V^{P}$}
\label{ged_general}
\end{figure}
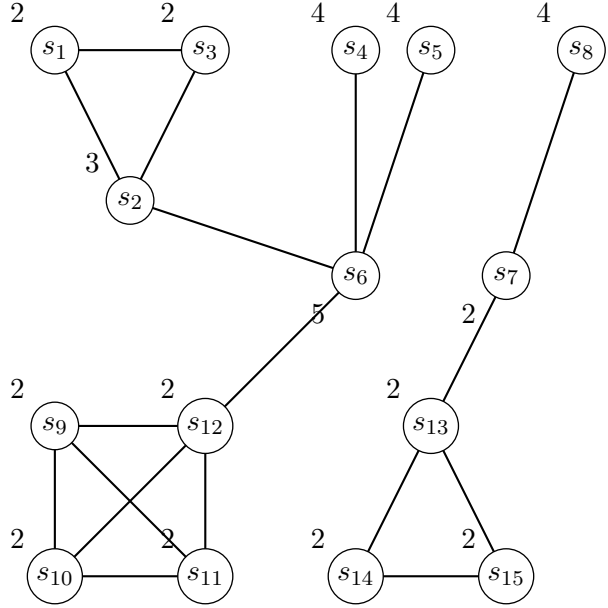

Let $I \subseteq V$ and $N(I) = \{j| ij \in E, i \in I\}$. Then $(I,N(I))$ is a reduced sub problem of the original problem. 

\begin{lemma}
Let $I= V \backslash V^{0}$ and let $\mu$ be a pareto-efficient matching for the original problem ($\mu(i)$ is the set of neighbors that is matched to an agent $i$), then
\begin{itemize}
\item For any agent $i \in V^{0}, \mu(i) \subseteq V^{U}$
\item For any even comoponent $(J,N({J}))$ such that $J \subseteq V^{P}$ and for any agent $i \in J, \mu(i) \in J\backslash i$
\item For any odd comoponent $J (\|J\| >= 2)$, the maximum size of b-matching within every odd comoponent is $\Sigma_{j=1}^{|J|} b_{j} - 1$. Moreover, for any odd comoponent $(J,R_{J})$, either
\begin{itemize}
\item one and only one agent $i \in J$ is matched with a agent in $V^{O}$ under the pareto efficient matching $\mu$ whereas all remaining agents in $J$ are matched within so that $\mu(j) \in J \backslash \{i,j\}$ for any agent $j \in J \backslash \{i\}$ or
\item one agent $i \in J$ remains unsaturated under the Pareto-efficient matching $\mu$ whereas all remaining agents in $J$ are saturated so that $mu(j) \in J \backslash \{i,j\}$ for any agent $j \in J \backslash {i}$
\end{itemize}

\end{itemize}

\end{lemma}

\begin{proof}
The Gallai-Edmonds decomposition precisely proves the above lemma when all the peaks $b_{j} = 1$. Roth et al.~\cite{roth2005pairwise} which this paper builds upon has a statement when the nodes have unit capacity. Now, lets construct a matching instance of the given b-matching problem. 
From the given graph $(G,V,E)$ construct the following graph $(G',V',E')$:
Create $b_{i}$ duplicate copies of node $i$ with unit peak each. Label these nodes as $(i^{1},i^{2}....i^{b_{i}})$.  In the graph $G'$, nodes $i^{k}, (k \in 1,2,..b_{i} )$ and $j^{l}, (l \in 1,2,..b_{j} )$ is connected by an edge if $ij$ is conencted in the original graph $G$. There is no edge between $i^{k}$ and $i^{r}, (k,r \in (1,2,...b_{i}))$ in $G'$. Now, we have a graph with unit peaks. Applying the GED on this unit graph would establish the result. 

We just stress a bit on the third statement.  When $(|J| = 1)$ the odd component has only 1 element, then there is no b-matching within that component. When $|J|>=2$, we highlight that the maximum size of any matching within every odd component in a unit capacity graph is one less than the number of nodes in the component (refer to Roth et al.~\cite{roth2005pairwise} for a proof). By symmetry, all the duplicate (identical) copies of a node will be in the same component. Hence, we can shrink these duplicate nodes into the parent node with original peak. Thus, from the decomposition above, the sum of the peaks of the agents in a odd component is $\Sigma_{j=1}^{|J|} b_{j}$ and the result follows. \qed
\end{proof}

\section*{\underline{Allocation Rules}}
\subsection*{\underline{Utility Profile:}} A utility profile $U \in R^{V}$ is said to be a feasible utility profile if the profile is an outcome of a feasible b-matching in the network. Since, we are looking for pareto optimal solutions, the outcome of our mechanisms is an utility profile induced by a maximum b-matching in the network.

A mechanism is \emph{deterministic} if for a given network, it picks a maximum b-matching as an outcome from the set of maximum b-matchings. Roth et al.~\cite{roth2005pairwise} study priority mechanisms in the context of deterministic mechanisms. They also show establish that the randomized "Egalitarian"  mechanism is superior in the sense that it not only retains the efficiency, strategyproof properties of deterministic mechanism, but it also produces an outcome which is envy free and strongly efficient in the sense of lorenz dominance. The mechanism of Roth et al. is a "lottery" mechanism, that randomly picks a maximum b-matching from the pareto optimal set. The distribution or lottery is chosen in such a way that the aforementioned economic properties are exhibited by the mechanism. We define this more mathematically below.

\subsection*{\underline{Lottery Mechanism:}} Let $\mathcal{U}$ be the set of all b-matchings in $G$. A matching lottery $l:$ $P_{\mu}, \mu \in M$ is a probability distribution over $\mathcal{U}$. For each matching $\mu \in \mathcal{U}$, $P_{\mu}, \mu \in [0,1]$ is the probability of choosing matching $\mu$ in lottery l, and $\sum_{\mu \in \mathcal{U}} P_{\mu} = 1$. A matching lottery $l$ can be also viewed as a fractional b-matching which

is defined to be a convex combination of several integral b-matchings. Let $\mathcal{L}$ be the set of matching lotteries. Given $l \in \mathcal{L}$, define the utility $x^{l}_{i}$ of vertex i to be the expected total exchanges that i is involved , i.e., $x^{l}_{i} = \sum_{\mu \in \mathcal{U}} l_{\mu}x_{i}^{\mu}$. Define the utility profile induced by lottery $l$ to be the vector $x_{l} = (x^{l}_{i}), i \in V$ . Let $P = {x_{l}}, l \in \mathcal{L}$ be the set of all feasible utility profiles.

\subsection*{\underline{Egalitarian Mechanism: Description and Properties}}

It is now clear from our definition that a feasible utility profile can be understood as a fractional b-matching (a convex combination of integral b-matchings). The $P_{rank}$ polymatroid we defined earlier is exactly the set of feasible utility profiles.

Roth et al.~\cite{roth2005pairwise} described the Egalitarian Mechanism for the pairwise kidney exchange problem. Jianli et al.~\cite{jianli_kidneylp} develop a water filling algorithm which gives a simpler and intuitive proof of the Roth's mechanism. Here, we use the ideas from jianli et al. to develop an Egalitarian Mechanism for the generalized maximum indivisible exchange problem. The idea is similar in the sense that the problem on a general network can be suitably reduced to a problem on a bipartite network. This can be done if one could establish the equivalence between the solutions in both network. Once the above step is completed, we could run the familiar allocation rules for bipartite networks to establish a allocation for our original general network. We describe this procedure below in more detail\\

\noindent \underline{\textbf{Step 1:}} \emph{Transformation to equivalent bipartite network:}
We construct the following two sided bipartite flow network. The left side of the bipartite network consists of nodes for each agent i. \\

Construction of Nodes: Define $\mathcal{C} = \{C_{1},C_{2}....C_{k}\}$ as the set of odd components in the underdemanded component $V^{U}$ where $k = |C|$. Construct the following bipartite graph $G_{B} = (A,B,E)$ where each node in $A$ corresponds to a node in $V$. We use $A^{C_{1}},...,A^{C_{k}}$ to denote $k$ disjoint sets corresponding to $C_{1},C_{2},..C_{k}$ respectively. Let $A^{PO} = \{a_{i} | i \in V^{P} \cup N^{O} \}$; $A^{U} = \{a_{i} | i \in V^{U}\}$. 

The construction of B is as follows. B Can be partitioned into 3 parts: $B^{PO} \cup B^{O} \cup B^{C}$. Each node in $B^{PO}$ corresponds to a node in $B^{P} \cup B^{O}$. We use $c_{i} \in B^{PO}$ to denote the node corresponding to $i \in B^{P} \cup B^{O}$. Each node in $B^{P}$ corresponds to a node $V^{P}$. We use $c_{i}' \in B^{O}$ to denote the node corresponding to $i \in V^{P}$. $B^{C}$ consists of $k$ disjoint sets $B^{C_{1}}, B^{C_{2}}....B^{C_{k}}$, where $B^{C_{i}}$ contains a node with value $\Sigma_{j \in C_{i}} b_{j} - 1 $ (this is the maximum b-matching within a particular odd component) if and only if $|B^{C_{i}}| \geq 2$; $B^{C_{i}}$ is empty otherwise.\\

Construction of Edges: \\
(1) For each $i \in V^{P} \cup V^{O}$, we have edge $(a_{i},c_{i}) \in E$ where $a_{i} \in A^{PO}$ and $c_{i} \in B^{PO}$ \\
(2) For a vertex $a_{i} \in A^{C}$ and another vertex $c'_{j} \in B^{O}$ we have $(a_{i},c'_{j}) \in E$ if $(i,j) \in E$ \\
(3) For each vertex, $a \in A^{C_{i}}$ add an edge to the node in  $B^{C_{i}}$.
\\ \\
\noindent \underline{\textbf{Step 2:}} \emph{Apply the egalitarian mechanism (BIMS)}:
Apply the Egalitarian Mechanism for indivisible goods in a bipartite network. \\

\noindent \underline{\textbf{Step 3:}} \emph{Construct a randomized mechanism: lottery over b-matchings:} The egalitarian mechanism outcomes a fractional utility profile for the agents. From theorem~\ref{thm:lott} later, it is clear that such a utility profile is actually feasible and can be generated as a lottery over integral maximum b-matchings. Every time, we generate a maximum b-matching with this probability profile and conduct exchanges on the network.
\qed
\\ \\
\noindent \textbf{Example:} Consider the original network as given in figure~\ref{ged_general}, the bipartite transformation outlined in step 1 is shown in 
figure~\ref{bip_ged}. When the egalitarian mechanism is applied to this bipartite network, we get the following utility profile for the agents: Every agent $i \in V^{P} \cup V^{O}$ receive their allocation = peak. The agents $s_{1},s_{2},s_{3}$ receive 2 units of utility each. 
The utility of agents $s_{2},s_{4},s_{5}$ is $\frac{7}{3}$ each. This fractional utility is obtained as a lottery over integral b-matchings. In this case, 2 units of agents $s_{4},s_{5}$ is always exchanged with agent $s_{6}$. $s_{2}$ has to exchange 2 units within its odd component in every maximum b-matching. Agents $s_{2},s_{4},s_{5}$ compete for the extra unit of exchange that $s_{6}$ can do. The egalitarian mechanism picks any of the possibility with a probability of $\frac{1}{3}$.

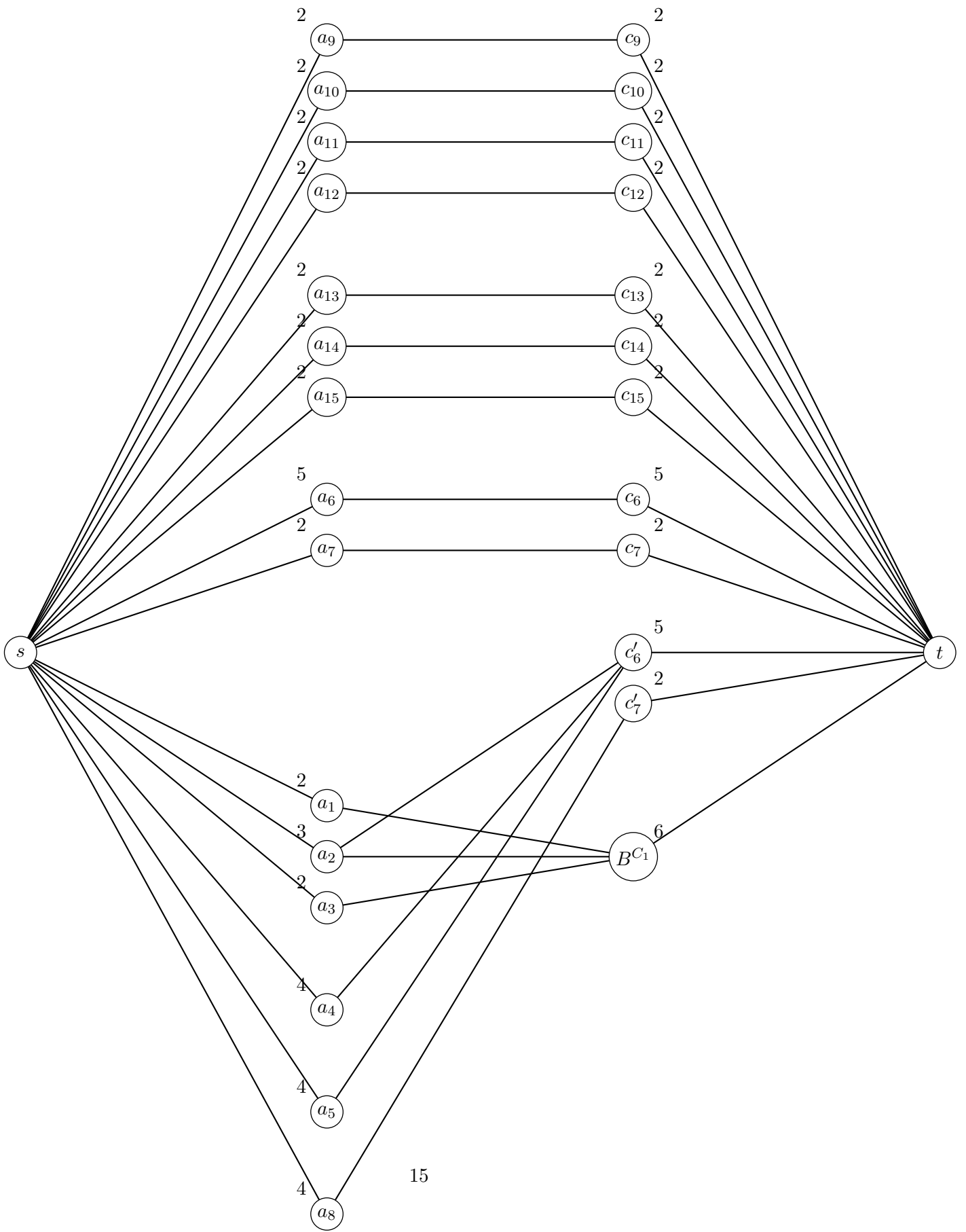
\begin{figure}
\begin{center}
\begin{tikzpicture}

\Vertex[x=0,y=12]{$a_9$}
\node at (-0.5,12.5){$2$};

\Vertex[x=0,y=11]{$a_{10}$}
\node at (-0.5,11.5){$2$};

\Vertex[x=0,y=10]{$a_{11}$}
\node at (-0.5,10.5){$2$};

\Vertex[x=0,y=9]{$a_{12}$}
\node at (-0.5,9.5){$2$};

\Vertex[x=6,y=12]{$c_9$}
\node at (6.5,12.5){$2$};

\Vertex[x=6,y=11]{$c_{10}$}
\node at (6.5,11.5){$2$};

\Vertex[x=6,y=10]{$c_{11}$}
\node at (6.5,10.5){$2$};

\Vertex[x=6,y=9]{$c_{12}$}
\node at (6.5,9.5){$2$};

\Vertex[x=0,y=7]{$a_{13}$}
\node at (-0.5,7.5){$2$};

\Vertex[x=0,y=6]{$a_{14}$}
\node at (-0.5,6.5){$2$};

\Vertex[x=0,y=5]{$a_{15}$}
\node at (-0.5,5.5){$2$};

\Vertex[x=6,y=7]{$c_{13}$}
\node at (6.5,7.5){$2$};
\Vertex[x=6,y=6]{$c_{14}$}
\node at (6.5,6.5){$2$};
\Vertex[x=6,y=5]{$c_{15}$}
\node at (6.5,5.5){$2$};

\Vertex[x=0,y=3]{$a_6$}
\node at (-0.5,3.5){$5$};
\Vertex[x=0,y=2]{$a_7$}
\node at (-0.5,2.5){$2$};

\Vertex[x=6,y=3]{$c_6$}
\node at (6.5,3.5){$5$};
\Vertex[x=6,y=2]{$c_7$}
\node at (6.5,2.5){$2$};

\Vertex[x=6,y=0]{$c'_6$}
\node at (6.5,.5){$5$};
\Vertex[x=6,y=-1]{$c'_7$}
\node at (6.5,-0.5){$2$};

\Vertex[x=-6,y=0]{$s$}
\Vertex[x= 12,y=0]{$t$}

\Vertex[x=0,y=-3]{$a_1$}
\node at (-0.5,-2.5){$2$};
\Vertex[x=0,y=-4]{$a_2$}
\node at (-0.5,-3.5){$3$};
\Vertex[x=0,y=-5]{$a_3$}
\node at (-0.5,-4.5){$2$};

\Vertex[x=6,y=-4]{$B^{C_{1}}$}
\node at (6.5,-3.5){$6$};

\Vertex[x=0,y=-7]{$a_4$}
\node at (-0.5,-6.5){$4$};

\Vertex[x=0,y=-9]{$a_5$}
\node at (-0.5,-8.5){$4$};

%

\Vertex[x=0,y=-11]{$a_8$}
\node at (-0.5,-10.5){$4$};


%

\Edges($a_9$,$c_9$)
\Edges($a_{10}$,$c_{10}$)
\Edges($a_{11}$,$c_{11}$)
\Edges($a_{12}$,$c_{12}$)
\Edges($a_{13}$,$c_{13}$)
\Edges($a_{14}$,$c_{14}$)
\Edges($a_{15}$,$c_{15}$)

\Edges($a_6$,$c_6$)
\Edges($a_7$,$c_7$)

\Edges($a_2$,$c'_6$)
\Edges($a_4$,$c'_6$)
\Edges($a_5$,$c'_6$)

\Edges($a_8$,$c'_7$)

\Edges($a_1$,$B^{C_{1}}$)
\Edges($a_2$,$B^{C_{1}}$)
\Edges($a_3$,$B^{C_{1}}$)


%
%

\Edges($s$,$a_1$)
\Edges($s$,$a_2$)
\Edges($s$,$a_3$)
\Edges($s$,$a_4$)
\Edges($s$,$a_5$)
\Edges($s$,$a_6$)
\Edges($s$,$a_7$)
\Edges($s$,$a_8$)
\Edges($s$,$a_9$)
\Edges($s$,$a_{10}$)
\Edges($s$,$a_{11}$)
\Edges($s$,$a_{12}$)
\Edges($s$,$a_{13}$)
\Edges($s$,$a_{14}$)
\Edges($s$,$a_{15}$)

\Edges($t$,$c_9$)
\Edges($t$,$c_{10}$)
\Edges($t$,$c_{11}$)
\Edges($t$,$c_{12}$)
\Edges($t$,$c_{13}$)
\Edges($t$,$c_{14}$)
\Edges($t$,$c_{15}$)
\Edges($t$,$c_6$)
\Edges($t$,$c_7$)
\Edges($t$,$c'_6$)
\Edges($t$,$c'_7$)
\Edges($t$,$B^{C_{1}}$)


\end{tikzpicture}
\end{center}
\caption{Transformation into bipartite network of the original graph in figure~\ref{ged:general}; $a_{i}$ refers to node $s_{i}$}
\label{bip_ged}
\end{figure}

\newpage
\subsection*{\underline{Egalitarian Mechanism: Linear Programming Approach}}

1) Step 1. Solve the linear program $ LP_{1}$
\begin{eqnarray*}
 &  \text{Maximize} \hspace{2mm} \lambda_{1} \\
& \text{subject to} \\
& x_{v} = \lambda_{1}, \forall v \in V, x \in P_{rank} 
\end{eqnarray*}
We call an element $v \in V$ a tight element if $x_{v}$ participates in some tight constraint in $P_{rank}$. Let $D_{1}$ be the set of tight elements.  In other words, increasing $x_{v}$ would violate some constraint in $P_{rank}$ when other $x_{u}$ for $u \neq v$ are fixed. In many linear programming algorithms, we can easily detect such tight elements. Another way to test the tightness of an element is to solve a closely related linear program in which we fix other $x_{u}$ and maximize $x_{v}$.

(2) In general, at Step $k$, we solve the linear program 
\begin{eqnarray*}
& LP_{k} = \text{maximize} \hspace{2mm}  \lambda_{k}, \text{subject to} \hspace{2mm} x_{v} = \lambda_{j}, \hspace{2mm} \forall \hspace{2mm} v \in D_{j} \hspace{2mm} \forall \hspace{2mm} j < k, \\
& x_{v} =\lambda_{k}, \forall v \in V \backslash \cup_{j < k} D_{j},\hspace{2mm} x \in P_{rank}. 
\end{eqnarray*}
Let $D_{k}$ be the set of elements that become tight in this step. \\

(3) The algorithm return $x={x_{v} =\lambda_{j} \hspace{2mm} \text{for} \hspace{2mm} v \in D_{j}} \hspace{2mm} \text{if} \hspace{2mm} \cup_{j=1}^{k} D_{j} =V$ \qed

\begin{lemma}
 Every feasible maximum flow on the modified bipartite network is equivalent to a feaible utility profile in the original network
\end{lemma}

\textbf{Proof:} A utility profile in the original network is induced by a maximum b-matching $\mu$ or a lottery (linear combination) over the set of maximum b-matchings. Also, it is well known that the polytope of maximum flows is convex with integral extreme points. Hence, to prove that every utility profile in original network implies a feasible maximum flow in the modified network, it is sufficient to prove the statement only for integral utility profiles (as every deterministic maximum b-matching corresponds to a integral utility profile). Now, given this maximum b-matching, we will match it to a unique maximum flow, $f$ in the modified network. Since each agent $i \in N^{P} \cup N^{O}$ is saturated, send flow $f_{ij} = b_{i}$ on all edges $ij, i \in A^{PO}, j \in B^{PO}$. Now, focus on each component $C_{k}$, if $C_{k}$ has only one node $i$ with utility $U_{i}$. Then this utility is derived only by exchanging with the agents in $B^{O'}$. Send a flow, $f_{ij} = x_{ij} \hspace{2mm} \forall ij$ such that $j \in N(i) \cap B^{O'} $ where $x_{ij}$ is the number of units exchanged between agents $i$ and $j$ in the given maximum b-matching. On a similar note, we fix the flow for the other components $C_{k}$. We know in any $C_{k}$, $|C_{k}|-1$ of the vertices are matched in the same component, so a flow of $|C_{k}|-1$ can be obtained by sending $x_{ij}$ units of flow to the agents in $B^{C_{k}}$.  If $ij \in \mu, i \in V^{C}$ and $j \in V^{O}$, then node $i$ is saturated in that particular b-matching; If we  send a unit flow in modified network, it is also saturated in the bipartite network.

To prove the other direction that every integral maximum flow corresponds to a feasible utility profile in the original network. Consider a maximum flow of $|F|$ in the modified bipartite network. First lets say the maximum flow value is $|B|$, this is also the size of the maximum b-matching in the original network. For a given set $A^{C_{i}}$ with $|A^{C_{i}}|>=2$, there is at most one vertex $a_{j} \in A^{C_{i}}$ such that there exists a vertex $c'_{h} \in B^{O}$ with $f(a_{j},c'_{h}) = 1$.In this case, we include in the matching $\mu$. 
In any maximum flow, for each $C_{i}$, agents in $A^{C_{i}}$ send $|C_{i}| - 1$ units of flow to $B^{C_{i}}$. Consequently from GED lemma, $|C_{i}|-1$ vertices can be matched among themselves. There are exactly $|V^{O}|$ units exchanged by agents in $A^{C}$ with agents in $V^{O}$. each such exchnge constitutes one unit flow to $B^{O}$. So all vertices in $V^{O}$ are matched. From GED Lemma, we can match all vertices in $V^{P}$ among themselves in $\mu$. It is easy to see that $u$ is exactly the utility profile corresponding to $\mu$.

\begin{theorem}
\label{thm:lott}
 Once we have this, then we can find the lottery mechanism by solving an LP - linear combination of extreme points
\end{theorem}

\textbf{Proof:} From the lemma above, we have showed that every maximum flow in the modified network is equivalent to a maximum b-matching induced utility profile in the original network. The egalitarian allocation outputs a maximum flow allocation (not necessarily integral). But it is folklore that this non-integral maximum flow can be obtained as a convex combination of integral maximum flows. These convex combinations $\lambda$ are the probabilities with which we pick different matchings from the set of all maximum b-matchings. 



\subsection*{\underline{Lorenz Dominance $\&$ No Envy}}

As discussed in Klaus~\cite{ehlers2003probabilistic}, No envy and ETE are not equivalent here. From theorem above, the egalitarian mechanism applied to this bipartite network is a feasible utility profile. Bochet et al.~\cite{bims} have established the egalitarian mechanism is lorenz dominant among all feasible flows. No Envy  also follows from the allocation rule of Bochet et al.~\cite{bim,bims}

\subsection*{\underline{Extensions and Consistency}} 
\textbf{Extension of Rules:} In the language of Moulin and Sethuraman~\cite{ms}, an allocation rule $\phi$ on a bipartite network $(G,V,E)$ is said to be an extension of the Sprumont's Uniform rule if $\phi$ coincides with the allocation of uniform rule if $G$ is a network with unit demander (supplier) connected to multiple suppliers (demanders). i.e.
\begin{equation}
\phi_{i} = U_{i}, \hspace{2mm} \forall \hspace{2mm} i \in V
\end{equation}
where $U_{i}$ is the utility of agent $i$ under uniform rule.

The notion of "extending a rule" is such that we are not compromising on properties of existing fair allocation mechanisms. Instead, we are actually generalizing in a suitable way to more general network structures.

In this spirit, Bochet et al.~\cite{bim,bims}, Chandramouli and Sethuraman~\cite{cs,cs2} develop mechanisms which are extensions of the uniform rule. Moulin and Sethuraman~\cite{ms} study a more general class of extensions of some basic well known rules. We extend this definiton further to general non-bipartite networks. An allocation rule $x$ on a general network $(G,V,E)$ is said to be an extension of the BIMS Egalitarian rule $\phi$ if $x$ coincides with $\phi$ if the network $G$ is bipartite i.e. $x_{i} = \phi_{i}, \hspace{2mm} \forall \hspace{2mm} i \in V$

\begin{lemma}
The egalitarian rule described earlier for non-bipartite networks is a extension of the  probabilitic egalitarian rule (described in appendix) for bipartite networks
\end{lemma}

\textbf{Proof:} If the network is bipartite, the odd components set, $V^{U}$ forms an independent set. So, $V^{U}$ is a collection of nodes (with peaks $b_{i}$) which are not saturated in atleast one maximum b-matching. Hence, the set $B^{U}$ is empty and each agent $i \in V^{U}$ is only connected to agents in $V^{O}$ who are completely saturated. In the bipartite context, if $i$ were a supplier then, any agent $j \in N(i)$ has $x_{j} = b_{j}$ in all pareto allocations. For a detailed derivation of GED for bipartite networks, follow the discussion in Chandramouli and Sethuraman~\cite{cs}. It is clear from their proof that if the given network $G$ is bipartite, then the Step 1 of our algorithm which modifies the given network into a suitable bipartite network, outputs $G$. Since, we apply the probabilistic egalitarian rule in step 2, the result follows. \qed
\\ \\
\textbf{Consistency:}
As discussed in Chandramouli and Sethuraman~\cite{cs2}, the egalitarian rule is an extension of the uniform rule which is not consistent. They propose the edge fair rule which is a consistent extension of the uniform rule to bipartite networks. In the previous section, if after the transformation into bipartite network, if we apply the edge fair mechanism to that network, we would have a consistent extension of the uniform rule to non-bipartite indivisible goods network.

\section*{\underline{Strategic Issues}}

\subsection*{Peak Strategyproofness}
In the discussions so far, we have ignored the possibility of agents having control over the values $b_{i}$. Suppose agents report $b_{i}$ to the mechanism designer who then decides the final allocation, then the agents can misreport the peaks to improve their allocation.
So far, the allocation of the agents was strictly less than $b_{i}$. If the agents report $b_{i}' > b_{i}$, then there is a possibility of agents having an allocation more than his true $b_{i}$. In that case, we need assumptions on the preference profile of agents to compare his utilities when his allocations are on either side of his true peak $b_{i}$.

Single peaked preferences: In the spirit of Bochet et al.~\cite{bim,bims}, we would like to continue our assumption of single peaked preferences for the agent allocations. Mathematically, given a prefernce profile $R_{i}$ for agent $i$ and two possible allocations $x_{i},x_{i}'$ then:
\begin{eqnarray}
& x_{i}' < x_{i} \leq p[R_{i}] \implies x_{i}P_{i}x_{i}' \\
& p[R_{i}] \leq x_{i} < x_{i}' \implies x_{i}P_{i}x_{i}'
\end{eqnarray}

A mechanism is \emph{peak strategyproof} if it is a dominant strategy for the agents to reveal their peaks truthfully.
Given this preference structure, the following example shows that there is no peak strategyproof mechanisms in the set of pareto optimal allocations.  Consider agents $(a,b,c)$ connected to each other and peak $b_{i} = 1, i \in (a,b,c)$. If the agents report their true peaks, then any allocation mechanism is such that $\{(x_{a},x_{b},x_{c})|x_{a} + x_{b} + x_{c} \leq 2\}$. W.l.o.g, lets say $x_{a} < 1$. Suppose agent $a$ misreports his peak $b_{a} = 2$, then we have a unique maximum b-matching and the allocation is $(2,1,1)$. If $2P_{a}x_{a}$ for agent $a$, he improves his allocation.

\subsection*{\underline{Link Strategyproofness} }

A mechanism is \emph{link strategyproof} if it is dominant strategy for the agents to reveal all his/her neighbors. Roth et al.~\cite{roth2005pairwise} established the link strategyproofness of the egalitarian mechanism in the kidney exchange problem. The egalitarian mechanism is not link group strategyproof even in the kidney exchange problem. To see that, consider the following network where each node has 1 unit of good to exchange. Agent $s_{4}$ is matched in every maximum matching, but $s_{7}$ is missed in some maximum matching and recieves a utility strictly less than 1 in the egalitarian allocation. Now, $s_{4},s_{7}$ can coordinate and misreport about the existence of the link between $s_{4}$ and $s_{3}$. If that link is not reported, then, $(s_{4},s_{5},s_{6},s_{7})$ form a separate group and all the agents are matched and receive peak utility. Hence, agent $s_{7}$ improves his allocation. 
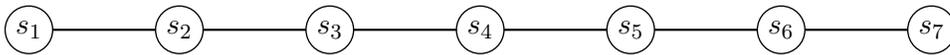
\begin{figure}[h!]
\begin{center}
\begin{tikzpicture}

\Vertex[x=-6,y=0]{$s_1$}
\Vertex[x=-4,y=0]{$s_2$}
\Vertex[x=-2,y=0]{$s_3$}
\Vertex[x=0,y=0]{$s_4$}
\Vertex[x=2,y=0]{$s_5$}
\Vertex[x=4,y=0]{$s_6$}
\Vertex[x=6,y=0]{$s_7$}

\Edges($s_1$,$s_2$)
\Edges($s_2$,$s_3$)
\Edges($s_3$,$s_4$)
\Edges($s_4$,$s_5$)
\Edges($s_5$,$s_6$)
\Edges($s_6$,$s_7$)

\end{tikzpicture}
\end{center}
\end{figure}

\noindent \textbf{Weak link groupstrategyproof:} A mechanism is weakly link group strategyproof if in a coalitional group, every agent who misreports about his connectivity strictly improves his/her allocation.

The agents in over demanded and perfectly demanded component doesn't misreport since they already recieve their peak allocation. The only deviating subsets are those agents in the odd components. Hence, for the rest of the proof we restrict our attention to coalition groups which consists only of agents of the odd components.

\begin{theorem}
The egalitarian mechanism is weakly link groupstrategyproof
\end{theorem}

{\bf Proof.}  We prove the result for an arbitrary coalition of agents. Let $A_i$ be the set of
agents
that agent $i$ is linked to, and let $A'_i$ be agent $i$'s report. We may
assume without loss of generality that any given agent $\in V^{PO}$ finds all the
agents in odd components
acceptable: if even component agent $j$ finds add component agent $i$ unacceptable, then agent $i$
cannot have a link to demander $j$ regardless of his report, so clearly $i$'s
manipulation opportunities are more restricted. Let $\phi$ and $\phi^{'}$ be
(any) egalitarian
flows when the suppliers report $A$ and $A'$ respectively, and let $x$ and $x'$
be
the corresponding allocation to the agents. We show that no coalition of
odd component agents
can weakly benefit by misreporting their links unless each agent in the
coalition
gets exactly their egalitarian allocation. 

Consider the transformed bipartite network of the given network when going through the analysis below:

The proof is by induction on the number of type 2 breakpoints in the algorithm
to compute
the egalitarian allocation in the transformed network. Suppose the given instance has $n$ type 2
breakpoints, and
suppose $X_1, X_2, \ldots, X_n$ are the corresponding bottleneck sets of
suppliers.
If $n = 0$, every odd component agent is at his peak value in the egalitarian allocation,
and
clearly this allocation cannot be improved. Suppose $n \geq 1$.
Define $$\tilde{X}_{\ell} = \{ i \in X_{\ell} \mid \sum_{j \in A_i} \phi'_{ij}
\; \geq \;
				\sum_{j \in A_i} \phi_{ij} \},$$
and
$$\hat{X}_{\ell} = \{ i \in X_{\ell} \mid \sum_{j \in A_i} \phi'_{ij} \; \leq \;
				\sum_{j \in A_i} \phi_{ij} \}.$$
We shall show, by induction on $\ell$, that for each $\ell = 1, 2, \ldots, n$:
\begin{itemize}
\item[(a)] $\phi^{'}_{ij} = 0$ for any $i \in \tilde{X}_{\ell'}$, $j \in
\cup_{i' \in X_{\ell}} A_{i'}$, $\ell' > \ell$; and
\item[(b)] $X_{\ell} \subseteq \hat{X}_{\ell}$.
\end{itemize}
The theorem follows from part~(b) above.

Any supplier $k \in X_{\ell} \setminus \tilde{X}_{\ell}$ must have
$A_k = A'_k$ as otherwise supplier $k$ is part of the deviating coalition and
does
worse. 
Consider now a supplier $i \in \tilde{X}_{\ell}$ with $x_i <
s_i$ and a 
supplier $k \in X_{\ell} \setminus \tilde{X}_{\ell}$. We have the following
chain of
inequalities:
$$\sum_{j \in A'_k} \phi^{'}_{kj} \; = \; \sum_{j \in A_k} \phi^{'}_{kj}
				\; < \; \sum_{j \in A_k} \phi_{kj} = x_k \; \leq
\; x_i 
				\; = \; \sum_{j \in A_i} \phi_{ij}
				\; \leq \; \sum_{j \in A_i} \phi^{'}_{ij}
				\; \leq \; \sum_{j \in A'_i} \phi^{'}_{ij}.$$
To see why, note that as $k \in X_{\ell} \setminus \tilde{X}_{\ell}$, the second
inequality is true by definition, and also $A_k = A'_k$ (justifying the first
equality). Also $k, i \in X_{\ell}$ and $x_i < s_i$, implies $x_k < s_i$, as
suppliers $k$
and $i$
both belong to the same bottleneck set and supplier $i$ is below his peak; this
justifies
the third inequality. The fourth and fifth inequalities follow from the fact
that
$i \in \tilde{X}_{\ell}$ and the fact that $\phi^{'}_{ij}$ must be zero for all
$j \in A_i \setminus A'_i$.
This chain of inequalities implies that $x'_k < x_k \leq s_k$ and $x'_k <
x'_i$. 
Therefore, when the suppliers report $A'$, supplier $k$ must be a member of 
an ``earlier'' bottleneck set than supplier $i$. An immediate consequence is
that
demanders in $A'_k = A_k$ do not receive any flow from supplier $i$ when the
report is $A'$. 

By the induction hypothesis, supplier $i \in X_{\ell}$ does not send any
flow to the demanders in \textbf{$\cup_{1 \leq i' \leq \ell-1} \cup_{k \in
X_{i^{'}} } A_{k'}$.}
Therefore 
$$\{ j \mid \phi^{'}_{ij} > 0, j \in A_i\} \subseteq \{ j \mid \phi_{ij} > 0, j
\in A_i \}.$$
This observation, along with the fact that every $i \in \tilde{X}_{\ell}$ weakly
improves,
and the fact that $X_{\ell}$ is a type 2 breakpoint implies that 
$\sum_{j \in A_i} \phi^{'}_{ij} = \sum_{j \in A_i} \phi_{ij}$,
establishing~(b). Furthermore, in such a solution,
every demander $j \in A_i$ for $i \in \tilde{X}_{\ell}$ must receive all his
flow from the suppliers in $\tilde{X}_{\ell}$.In particular, the demanders in
$X_{\ell}$ 
cannot receive
any flow from suppliers in $X_{\ell'}$ for $\ell' > \ell$, establishing~(a).
To complete the proof we need to establish the basis for the induction proof,
i.e., the case of $\ell = 1$. This, however, follows easily: it is easy to
verify that the set $X_1 \setminus \tilde{X}_1$ must be empty, so 
$X_1 = \tilde{X}_1$. As $X_1$ is a type 2 bottleneck set, it is not possible
for {\em every} member of $X_1$ to do weakly better unless the allocation
remains unchanged. Thus, both (a) and (b) follow.
\qed




\subsection{Arbitrary Capacity}



%



\section{Conclusion}

\newpage

\section{Appendix}

\subsection{Indivisible exchange on bipartite networks: Probabilistic Egalitarian Rule}
We fix a problem $(G,s,d)$ such that $s_{i},d_{j}>0$ for all $i,j$ (clearly
if $s_{i}=0$ or $d_{j}=0$ we can ignore supplier $i$ or demander $j$
altogether). We define independently our solution for the suppliers and for
the demanders.

The definition for suppliers is by induction on the number of agents $%
|S|+|D| $. Consider the parameterized capacity graph $\Gamma (\lambda
),\lambda \geq 0 $: the only difference between this graph and $\Gamma
(G,s,d)$ is that the capacity of the edge $\sigma i,i\in S_{-}$ is $\min
\{\lambda ,s_{i}\}$, which we denote $\lambda \wedge s_{i}$. (In particular,
the edge from $j$ to $\tau $ still has capacity $d_{j}$). We set $\alpha
(\lambda ) $ to be the maximal flow in $\Gamma (\lambda )$. Clearly $\alpha $
is a piecewise linear, weakly increasing, strictly increasing at $0$, and
concave function of $\lambda $, reaching its maximum when the total $\sigma $%
-$\tau $ flow is $d_{D_{+}}$. Moreover, each breakpoint is one of the $s_{i}$
(type 1), and/or is associated with a subset of suppliers $X$ such that%
\begin{equation}
\sum_{i\in X}\lambda \wedge s_{i}=\sum_{j\in f(X)}d_{j}  \label{2}
\end{equation}%
Then we say it is of type 2. In the former case the associated supplier
reaches his peak and so cannot send any more flow. In the latter case the
group of suppliers in $X$ is a \textit{bottleneck}, in the sense that they
are sending enough flow to satisfy the collective demand of the demanders in 
$f(X)$ and these are the only demanders they are connected to; any further
increase in flow from any supplier in $X$ would cause some demander in $f(X)$
to accept more than his peak demand.

If the given problem does not have any type-2 breakpoint, then the
egalitarian solution obtains by setting each supplier's allocation to his
peak value. Otherwise, let $\lambda ^{\ast }$ be the first type-2 breakpoint
of the max-flow function; by the max-flow min-cut theorem, for every subset $%
X$ satisfying (\ref{2}) at $\lambda ^{\ast }$ the cut $C^{1}=\{\sigma \}\cup
X\cup f(X)$ is a minimal cut in $\Gamma (\lambda ^{\ast })$ providing a
certificate of optimality for the maximum-flow in $\Gamma (\lambda ^{\ast })$%
. If there are several such cuts, we pick the one with the largest $X^{\ast
} $ (its existence is guaranteed by the usual supermodularity argument). The
egalitarian solution obtains by setting 
\begin{equation*}
x_{i}=\min \{\lambda ^{\ast },s_{i}\},\;\text{for}\;i\in X^{\ast
},\;y_{j}=d_{j},\;\text{for}\;j\in f(X^{\ast }),
\end{equation*}%
and assigning to other agents their egalitarian share in the reduced problem 
$(G(S\diagdown X^{\ast },D\diagdown f(X^{\ast })),s,d)$. That is, we
construct $\Gamma ^{S\diagdown X^{\ast },D\diagdown f(X^{\ast })}(\lambda )$
for $\lambda \geq 0$ by changing in $\Gamma (G(S\diagdown X^{\ast
},D\diagdown f(X^{\ast })),s,d)$ the capacity of the edge $\sigma i$ to $%
\lambda \wedge s_{i}$, and look for the first type-2 breakpoint $\lambda
^{\ast \ast }$ of the corresponding max-flow function. An important fact is
that $\lambda ^{\ast \ast }>\lambda ^{\ast }$. Indeed there exists a subset $%
X^{\ast \ast }$ of $S\diagdown X^{\ast }$ such that%
\begin{equation*}
\sum_{i\in X^{\ast \ast }}\lambda ^{\ast \ast }\wedge s_{i}=\sum_{j\in
f(X^{\ast \ast })\diagdown f(X^{\ast })}d_{j}
\end{equation*}%
If $\lambda ^{\ast \ast }\leq \lambda ^{\ast }$\ we can combine this with
equation (\ref{2}) at $X^{\ast }$ as follows%
\begin{equation*}
\sum_{i\in X^{\ast }\cup X^{\ast \ast }}\lambda ^{\ast }\wedge s_{i}\geq
\sum_{i\in X^{\ast }}\lambda ^{\ast }\wedge s_{i}+\sum_{i\in X^{\ast \ast
}}\lambda ^{\ast \ast }\wedge s_{i}=\sum_{j\in f(X^{\ast }\cup X^{\ast \ast
})}d_{j}
\end{equation*}%
contradicting our choice of $X^{\ast }$ as the largest subset of $S_{-}$
satisfying (\ref{2}) at $\lambda ^{\ast }$.

Once the $\lambda$ is obtained the probabilistic egalitarian rule is obtained similarly 
to the probabilistic uniform rule of Klaus et al. and Moulin []. The trick here is to randomize
at each bottleneck. Now, without loss of generality, let $\bar{X} = \{i \in X^{**}| p(R_{i}) \geq x_{\lambda} + 1 \} = 
 \{1,2....,\bar{n}\}$ and $\tilde{N} = \{i \in N | p(R_{i}) \leq x_{\lambda}\} = \{\bar{n}+1,...,n\}$. Then in the final allocation we obtain, each agent in $\tilde{N}$ receives his peak amount and each agent in $\bar{N}$ 
 receives either $x_{\lambda}$ or $x_{\lambda} + 1$. Note that for each $i \in \bar{N}, (x_{\lambda}+1)P_{i}x_{\lambda}$ and that exactly $\bar{n}(\lambda - x_{\lambda})$ agents in $\bar{N}$ can receive $x_{\lambda} + 1$. The randomized "lottery" places equal probability on all allocations where all agents in $\tilde{N}$ 
 receive their peak amounts, $\bar{n}(\lambda - x_{\lambda})$ agents in $\bar{N}$ receive $x_{\lambda} + 1$
 and the remaining agents in $\bar{N}$ receive $x_{\lambda}$. Hence, the utility profile is obtained by placing
 equal probabilities on exactly $()$ allocations. 
 
 If $p(R_{i} \leq x_{\lambda}$, then $U_{i}(R)(p(R_{i})) = 1$ and 
 if $p(R_{i}) \geq x_{\lambda} + 1$, then $U_{i}(R)(x_{\lambda}+1)= \lambda - x_{\lambda}$
 and $U_{i}(R)(x_{\lambda}) = 1 - (\lambda - x_{\lambda})$
The solution thus obtained recursively is the egalitarian allocation for the
suppliers. A similar construction works for demanders: We consider the
parameterized capacity graph $\Delta (\mu ),\mu \geq 0$, with the capacity
of the edge $\tau j,j\in D$ set to $\mu \wedge d_{j}$. We look for the first
type-2 breakpoint $\mu ^{\ast }$ of the maximal flow $\beta (\mu )$ of $%
\Delta (\mu )$, and for the largest subset of demanders $Y$ such that%
\begin{equation*}
\sum_{j\in Y}\mu \wedge d_{j}=\sum_{i\in g(Y)}s_{i}
\end{equation*}%
etc.. Combining these two egalitarian allocations yields the egalitarian
allocation $(x^{e},y^{e})\in 
\mathbb{R}
_{+}^{S\cup D}$ for the overall problem.\medskip


\newpage
\nocite{*}
\bibliographystyle{plain}
\bibliography{refs}

\end{document}